\documentclass[11pt]{article}
\usepackage[utf8]{inputenc}
\usepackage[T1]{fontenc}
\usepackage[margin=1in]{geometry}
\usepackage{amsmath,amssymb,amsthm}
\usepackage{enumitem}
\setlist{nosep,leftmargin=1.7em}
\usepackage{xcolor}
\definecolor{linkcol}{rgb}{0.10,0.15,0.55}
\definecolor{citecol}{rgb}{0.00,0.35,0.15}
\usepackage[colorlinks=true,linkcolor=linkcol,citecolor=citecol,urlcolor=linkcol,breaklinks=true]{hyperref}
\usepackage[capitalise,nameinlink,noabbrev]{cleveref}

\newtheorem{theorem}{Theorem}[section]
\newtheorem{lemma}[theorem]{Lemma}
\newtheorem{proposition}[theorem]{Proposition}

\DeclareMathOperator{\supp}{supp}
\newcommand{\ip}[2]{\langle #1,#2\rangle}
\newcommand{\XX}{\mathcal{X}}
\newcommand{\YY}{\mathcal{Y}}
\newcommand{\PP}{\mathcal{P}}
\newcommand{\Ent}{\mathrm{H}}
\newcommand{\etafix}{\eta_{\mathrm{fix}}}
\newcommand{\poly}{\mathrm{poly}(n)}
\newcommand{\bh}{\mathrm{h}}
\newcommand{\A}{\mathsf A}
\newcommand{\B}{\mathsf B}
\newcommand{\C}{\mathsf C}

\title{\vspace{-2.5em}An $O\big((5/3)^n\mathrm{poly}(n)\big)$ One-Sided Monte Carlo Algorithm for \textsc{Equal Subset Sum}}
\author{Lixi Ye\thanks{Institute for Interdisciplinary Information Sciences, Tsinghua University, Beijing, China.}\\
{\small\nolinkurl{ylx25@mails.tsinghua.edu.cn}}}
\date{\today}

\begin{document}
\renewcommand{\thefootnote}{\fnsymbol{footnote}}
\setcounter{footnote}{1}
\maketitle
\setcounter{footnote}{0}
\renewcommand{\thefootnote}{\arabic{footnote}}
\vspace{-3em}

\begin{abstract}\noindent\small
We give a randomised algorithm for \textsc{Equal Subset Sum} that, on $n$ arbitrary integers of at most $m\le2^n$ bits, runs in time $O\big((5/3)^n\mathrm{poly}(n)+n^2m\big)$, never outputs a non-solution, and outputs a solution with probability $1-2^{-\Omega(n)}$ whenever one exists.
\end{abstract}

\section{The problem and the algorithm}\label{sec:alg}

An instance of \textsc{Equal Subset Sum} is a vector $x\in\mathbb Z^{n}$ whose entries have at most $m$ bits; a \emph{solution} is a pair of distinct sets $S\ne T\subseteq[n]$ with $\sum_{i\in S}x_i=\sum_{i\in T}x_i$. For positive entries the two sets may be taken disjoint and nonempty, by cancelling their intersection, which is the classical formulation of Woeginger and Yu \cite{WY92}, who showed the problem to be NP-complete. The folklore meet-in-the-middle algorithm, in the style of Horowitz and Sahni \cite{HS74}, solves it in time $3^{n/2}\poly\le1.7321^n\poly$, and improving on this was raised as an open problem by Woeginger \cite{Woe08}. Mucha, Nederlof, Pawlewicz and Węgrzycki \cite{MNPW19} gave an $O(1.7088^n)$ time Monte Carlo algorithm, by transferring the representation technique of Howgrave-Graham and Joux \cite{HJ10} from average-case to worst-case inputs; Chen, Jin, Randolph and Servedio \cite{CJRS22} developed it for subset balancing problems in the average case; Randolph and Węgrzycki \cite{RW26} improved the worst-case bound to $O(1.7067^n)$; and Yamano and Shibuya \cite{YS26} to $O(1.6994^n)$ time and $O(1.5664^n)$ space, the fastest bound currently known. In \cite{MNPW19} a solution is written as a difference of two vectors of $\{0,1\}^n$, so that a matched pair automatically has its difference in $\{-1,0,1\}^n$; \cite{RW26} gains further representations by letting both lists take values in $\{0,1,2\}$, at the cost that a matched pair need not have that property, and recovers the valid pairs by a separate compatibility testing procedure \cite[Section 8]{RW26}. The construction below also uses lists over $\{0,1,2\}$, but assigns each coordinate one of three random labels that restrict the two lists there so that every matched pair has its difference in $\{-1,0,1\}^n$ by construction; it pays for the assignment rather than for compatibility testing. We obtain the bound $(5/3)^n\poly$, where $5/3=1.6666\ldots$ (\cref{thm:main}).

Below, $\poly$ denotes an unspecified polynomial in $n$, possibly different at each occurrence, $q(n)=n^{c_0}$ is one fixed polynomial with $c_0$ a sufficiently large absolute constant, and $n$ is assumed to exceed a suitable constant. For $D\in\{-1,0,1\}^n$ and $y\in\mathbb Z^n$ we write $\ip{D}{y}=\sum_iD_iy_i$, and $\supp(D)$ for the coordinates where $D_i\ne0$. A \emph{halving} of $[n]$ is an ordered partition $[n]=H_1\dot\cup H_2$ with $|H_1|=\lfloor n/2\rfloor$ and $|H_2|=\lceil n/2\rceil$; a uniform halving is drawn uniformly among these. A uniform prime of an interval $[K,2K]$ with $K=2^{\Theta(n)}$ is drawn by sampling uniform integers of the interval and testing each for primality in $\poly$ time, giving up after $q(n)$ samples; the primes have density $\Omega(1/n)$ there by the prime number theorem, so the sampler gives up with probability $2^{-\Omega(n)}$, and conditioned on returning, its prime is uniform. A procedure whose sampler gives up abandons that iteration; \textsc{ESS} outputs \textsf{NO} if that of \cref{lem:reduce} does. Every hash table below stores exact integer keys and resolves collisions by exact comparison; sorting instead gives the same bounds with no assumption on the hash function, at one further factor of $\poly$.

The exponent we are aiming for is a function of $n$ alone, so the entries must first be brought down to $O(n)$ bits. If $m>2^n$ then $n<\log_2m$, and exhaustive search over $\{-1,0,1\}^n$ costs $3^n\cdot O(nm)=O(m^{1+\log_23}\log_2m)$ bit operations, polynomial in the input length; we therefore assume $m\le2^n$ from now on. The reduction we use is that of \cite[Appendix A of the full version]{MNPW19}, restated as \cite[Section 2]{YS26}; we prove the form we need, which applies to arbitrary integers and which our algorithm both calls and checks. Each entry is replaced by its residue modulo a random prime $P$ of $4n+1$ bits, which preserves equalities of subset sums only as congruences modulo $P$. To turn the congruences back into equalities, $\lceil\log_2n\rceil$ auxiliary entries $P,2P,4P,\dots$ are appended, whose signed subset sums realise every multiple $rP$ with $|r|<n$ that can arise, so the leftover is always absorbed. The prime is random because the converse direction is the fragile one: a solution of the reduced instance need not be one of the original, and $P$ is chosen so that no such false pair exists except with probability $2^{-\Omega(n)}$. The algorithm checks the pair it recovers against the original instance before printing it, so that failure of this event costs a solution rather than a wrong output.

\begin{lemma}\label{lem:reduce}
Let $x\in\mathbb Z^n$ have entries of at most $m\le2^n$ bits, put $k=\lceil\log_2n\rceil$, let $P$ be a uniform prime of $[2^{4n},2^{4n+1}]$, and let $\widetilde x\in\mathbb Z_{\ge0}^{\,n+k}$ be given by $\widetilde x_i=x_i\bmod P$ for $i\le n$ and $\widetilde x_{n+j}=2^{\,j-1}P$ for $1\le j\le k$. Then $\widetilde x$ has entries of at most $4n+k+1$ bits and is computed in $O(n^2m+\poly)$ bit operations, and
\begin{enumerate}[label=(\alph*), ref=\thetheorem(\alph*)]
\item\label{lem:reduce-a} if $x$ has a solution, then so has $\widetilde x$;
\item\label{lem:reduce-b} with probability $1-2^{-\Omega(n)}$ over $P$, the pair $(S\cap[n],\,T\cap[n])$ is a solution of $x$ for every solution $(S,T)$ of $\widetilde x$.
\end{enumerate}
\end{lemma}

\begin{proof}
Every entry is smaller than $2^{k}P<2^{4n+k+1}$. Drawing $P$ costs $\poly$, and dividing an $m$-bit integer by a $(4n+1)$-bit one costs $O(nm)$ bit operations by long division, so $\widetilde x$ costs $O(n^2m+\poly)$ in all; fast division would give $O\big(n^2+nm\log^{O(1)}m\big)$, near-linear in the combined input and output length. Write $\Sigma_A=\sum_{i\in A}x_i$ and $\widetilde\Sigma_A=\sum_{i\in A}\widetilde x_i$.

(a) Let $S\ne T$ be subsets of $[n]$ with $\Sigma_S=\Sigma_T$. Then $\widetilde\Sigma_S-\widetilde\Sigma_T\equiv\Sigma_S-\Sigma_T=0\pmod P$, while both terms lie in $[0,nP)$, so $\widetilde\Sigma_S-\widetilde\Sigma_T=rP$ for an integer $|r|\le n-1\le2^{k}-1$. Write $|r|=\sum_{j\in A}2^{\,j-1}$ for a set $A\subseteq\{1,\dots,k\}$ and adjoin $\{n+j:j\in A\}$ to $T$ if $r>0$ and to $S$ if $r<0$. The two sets still differ inside $[n]$, so are distinct, and their $\widetilde x$-sums now agree.

(b) Call $P$ \emph{good} if it divides no nonzero integer of the form $\Sigma_A-\Sigma_B$ with $A,B\subseteq[n]$. Each such integer is $\ip{D}{x}$ for the vector $D\in\{-1,0,1\}^n$ that is $+1$ on $A\setminus B$ and $-1$ on $B\setminus A$, so at most $3^n$ values arise, each of absolute value below $n2^m$ and hence with at most $(m+\log_2n)/4n$ prime divisors that are $2^{4n}$ or larger. As $[2^{4n},2^{4n+1}]$ contains at least $2^{4n}/4n$ primes, a fixed value is divisible by $P$ with probability at most $(m+\log_2n)2^{-4n}\le2^{\,n+1-4n}$, using $m\le2^n$; so $P$ is bad with probability at most $3^n2^{\,n+1-4n}=2^{-\Omega(n)}$.

Let now $P$ be good and let $(S,T)$ be a solution of $\widetilde x$, and put $S_0=S\cap[n]$ and $T_0=T\cap[n]$. The auxiliary entries are $P$ times distinct powers of two, so distinct subsets of them have distinct sums; were $S_0=T_0$, the sets $S\setminus[n]$ and $T\setminus[n]$ would have equal $\widetilde x$-sums and $S=T$ would follow. Hence $S_0\ne T_0$. Moreover $\widetilde\Sigma_{S_0}-\widetilde\Sigma_{T_0}$ is the $\widetilde x$-sum of $T\setminus[n]$ minus that of $S\setminus[n]$, a multiple of $P$, and $\Sigma_{S_0}-\Sigma_{T_0}\equiv\widetilde\Sigma_{S_0}-\widetilde\Sigma_{T_0}\pmod P$, so $\Sigma_{S_0}=\Sigma_{T_0}$ because $P$ is good.
\end{proof}

By \cref{lem:reduce} the search may be carried out on $n+O(\log n)$ nonnegative integers of $O(n)$ bits, on which an arithmetic operation costs $\poly$ time. Since $(5/3)^{n+O(\log n)}=(5/3)^n\poly$, the exponential factor is unaffected, and the reduction contributes the additive $O(n^2m)$ of \cref{thm:main}. Accordingly \textsc{Core} and its analysis are stated for $n$ integers of $O(n)$ bits, and \textsc{ESS} applies it to $\widetilde x$.

\begin{lemma}\label{lem:enc}
Let $x\in\mathbb Z^{n}$.
\begin{enumerate}[label=(\alph*), ref=\thetheorem(\alph*)]
\item\label{lem:enc-a} If $D\in\{-1,0,1\}^n$ is nonzero and $\ip{D}{x}=0$, then $\{i:D_i=1\}$ and $\{i:D_i=-1\}$ are a solution of $x$; conversely, if $S\ne T$ have equal sums, then the vector that is $+1$ on $S\setminus T$, $-1$ on $T\setminus S$ and $0$ elsewhere is nonzero and annihilates $x$.
\item\label{lem:enc-b} If $x\in\mathbb Z_{>0}^n$, the two sets produced in (a) are moreover nonempty.
\end{enumerate}
\end{lemma}

\begin{proof}
(a) The two sets are disjoint, they are distinct because $D\ne0$, and their difference of sums is $\ip{D}{x}=0$. Conversely the displayed vector is nonzero because $S\ne T$, and cancelling $\sum_{i\in S\cap T}x_i$ turns $\sum_{i\in S}x_i=\sum_{i\in T}x_i$ into $\ip{D}{x}=0$. (b) If one of the two disjoint sets were empty, the other would have sum $0$, which forces it to be empty as well because $x_i>0$, i.e. $D=0$.
\end{proof}

\begin{lemma}\label{lem:sign}
Let $\varepsilon\in\{\pm1\}^n$ and $y_i=\varepsilon_ix_i$. Then $D\mapsto(\varepsilon_iD_i)_i$ is a bijection, preserving supports, from $\{D\in\{-1,0,1\}^n\setminus\{0\}:\ip{D}{x}=0\}$ onto $\{D\in\{-1,0,1\}^n\setminus\{0\}:\ip{D}{y}=0\}$.
\end{lemma}

\begin{proof}
$\sum_i\varepsilon_iD_i\cdot\varepsilon_ix_i=\ip{D}{x}$, the map is an involution, and $\varepsilon_i\ne0$.
\end{proof}

This re-randomisation, and the reduction to a solution of minimum support with balanced signs that it supports in \cref{ssec:balanced}, are those of \cite[Section 4.1]{RW26}. The algorithm is driven by the number $z$ of zero coordinates of such a solution, whose density $p=z/n$ selects one of three branches; guessing this density, and treating solutions of unbalanced profile separately, follows \cite{MNPW19,RW26,YS26}.

\smallskip
\noindent\fbox{\begin{minipage}{0.945\textwidth}\small
\textbf{ESS$(x)$.} Compute the instance $\widetilde x$ of \cref{lem:reduce} and run \textsc{Core}$(\widetilde x)$. If it returns a vector whose restriction $D$ to $[n]$ satisfies $D\ne0$ and $\ip{D}{x}=0$, output the pair $\big(\{i:D_i=1\},\{i:D_i=-1\}\big)$ and halt; otherwise output \textsf{NO}.
\end{minipage}}

\smallskip
\noindent\fbox{\begin{minipage}{0.945\textwidth}\small
\textbf{\textsc{Core}$(x)$.} For $z=0,1,\dots,n-1$, put $p:=z/n$ and run exactly one branch:
\[ p\le\tfrac17\ \text{or}\ p\ge\tfrac23:\ \textsc{Mitm}(x,z);\qquad \tfrac17<p<\tfrac27:\ \textsc{Central}(x,z);\qquad \tfrac27\le p<\tfrac23:\ \textsc{Collide}(x,z). \]
Whenever a branch returns a vector $D$, check that $D\ne0$ and $\ip{D}{x}=0$; if the check passes, return $D$ and halt. If every $z$ fails, return \textsf{NO}.
\end{minipage}}

\smallskip
For an alphabet $\Sigma$ and a length $N$, a \emph{type} is a vector $\nu\in\mathbb Z_{\ge0}^{\Sigma}$ with $\sum_\sigma\nu(\sigma)=N$; a string of length $N$ \emph{has type $\nu$} if each $\sigma$ occurs exactly $\nu(\sigma)$ times, and the number of such strings is the multinomial coefficient $\binom{N}{\nu}$.

\smallskip
\noindent\fbox{\begin{minipage}{0.945\textwidth}\small
\textbf{\textsc{Mitm}$(x,z)$.} Repeat $q(n)$ times: draw a uniform halving $[n]=H_1\dot\cup H_2$. For each of the $\poly$ pairs $(\nu^1,\nu^2)$ of types on $\{-1,0,1\}$, of lengths $|H_1|$ and $|H_2|$, that agree entrywise up to $1$ and satisfy $\nu^1(0)+\nu^2(0)=z$: hash the keys $\ip{D^1}{x|_{H_1}}$ over all $D^1\in\{-1,0,1\}^{H_1}$ of type $\nu^1$, then query the key $-\ip{D^2}{x|_{H_2}}$ for every $D^2\in\{-1,0,1\}^{H_2}$ of type $\nu^2$, and return any hit $D=(D^1,D^2)\ne0$.
\end{minipage}}

\smallskip
\noindent\fbox{\begin{minipage}{0.945\textwidth}\small
\textbf{\textsc{Collide}$(x,z)$.} Let $w=\min\{z,\,n-z-1\}$ and fix any halving $H_1,H_2$. Repeat $q(n)$ times: draw a uniform prime $P\in[2^{w},2^{w+1}]$ and a uniform residue $\rho$ modulo $P$; bucket every $S\subseteq H_h$ by $\sum_{i\in S}x_i \bmod P$, storing only nonempty buckets; compute $N=\sum_r|\mathcal B^1_r|\,|\mathcal B^2_{(\rho-r)\bmod P}|$ and \textbf{abort} the iteration if $N>q(n)2^{\,n-w}$; otherwise materialise $\{S\subseteq[n]:\sum_{i\in S}x_i\equiv\rho\}$ by joining complementary buckets, hash its members by their exact sum, and on a collision $S\ne S'$ return the vector that is $+1$ on $S\setminus S'$, $-1$ on $S'\setminus S$ and $0$ elsewhere.
\end{minipage}}

\smallskip
For the third branch, fix $p=z/n\in(\tfrac17,\tfrac27)$ and set
\begin{equation}\label{eq:par}
\alpha=\frac{9p-1}{2},\qquad \beta=\frac{3(1-3p)}{4},\qquad \kappa=\frac{\alpha}{3},\qquad \etafix=2\log_23-\frac{16}{9},
\end{equation}
so that $\alpha,\beta>0$ and $\alpha+2\beta=1$. Write $\bh(p)=-p\log_2p-(1-p)\log_2(1-p)$, with the convention $0\log_20=0$ in force throughout, and let
\begin{equation}\label{eq:theta}
\Theta=\begin{array}{c|ccc} & \A & \B & \C\\\hline -1 & \kappa & 4\beta/9 & 4\beta/9\\ \ \ 0 & \kappa & \ \beta/9 & \ \beta/9\\ +1 & \kappa & 4\beta/9 & 4\beta/9 \end{array}
\end{equation}
be an array indexed by $\{-1,0,1\}\times\{\A,\B,\C\}$. Its column sums are $3\kappa=\alpha$, $\beta$ and $\beta$, hence $\alpha+2\beta=1$ in total; its row sums are $\kappa+\tfrac{8\beta}9=\tfrac{9p-1}6+\tfrac{2(1-3p)}3=\tfrac{1-p}2$ for the rows $\pm1$ and $\kappa+\tfrac{2\beta}9=\tfrac{9p-1}6+\tfrac{1-3p}6=p$ for the row $0$. So $\Theta$ is a probability distribution with column marginal $(\alpha,\beta,\beta)$ and row marginal $\big(\tfrac{1-p}2,p,\tfrac{1-p}2\big)$, and it is invariant both under interchanging the columns $\B$ and $\C$ and under interchanging the rows $+1$ and $-1$.

Coordinates carry a label in $\{\A,\B,\C\}$ and a half index in $\{1,2\}$, so a run of the branch is organised around an ordered partition $\Pi=(\Pi_{L,h})_{L,h}$ of $[n]$ into six blocks, with $H_h=\Pi_{\A,h}\cup\Pi_{\B,h}\cup\Pi_{\C,h}$. A \emph{layout} is a family $\nu=(\nu_{L,h})_{L,h}$ of types on $\{-1,0,1\}$ with $\sum_{L,h}\sum_d\nu_{L,h}(d)=n$ and $\sum_{L,h}\nu_{L,h}(0)=z$; it is \emph{admissible} if
\begin{equation}\label{eq:adm}
\Big|\nu_{L,h}(d)-\tfrac n2\,\Theta(d,L)\Big|\le4\qquad\text{for all }d,L,h .
\end{equation}
For each $z$ there are $O(1)$ admissible layouts and \textsc{Central} enumerates all of them. A partition $\Pi$ is \emph{compatible} with $\nu$ if $|\Pi_{L,h}|=\sum_d\nu_{L,h}(d)$ for all $L,h$. Given a layout, put
\begin{equation}\label{eq:gK}
K=\prod_{h=1,2}\binom{\nu_{\A,h}(0)}{\lfloor\nu_{\A,h}(0)/2\rfloor},
\end{equation}
and given in addition a compatible $\Pi$, let $\XX(\Pi,\nu)$ be the set of all $X\in\{0,1,2\}^n$ such that for $h=1,2$,
\begin{equation}\label{eq:lists}
\begin{aligned}
&X\ \text{restricted to}\ \Pi_{\A,h}\ \text{has type}\ \big(\nu_{\A,h}(-1)+\lfloor\nu_{\A,h}(0)/2\rfloor,\ \ \nu_{\A,h}(+1)+\lceil\nu_{\A,h}(0)/2\rceil,\ \ 0\big),\\
&X\ \text{restricted to}\ \Pi_{\B,h}\ \text{has type}\ \big(\nu_{\B,h}(-1),\ \nu_{\B,h}(0),\ \nu_{\B,h}(+1)\big),\\
&X\ \text{is constant}\ 1\ \text{on}\ \Pi_{\C,h},
\end{aligned}
\end{equation}
where types on $\{0,1,2\}$ are written as triples. Let $\widetilde\Pi$ be $\Pi$ with the blocks $\Pi_{\B,h}$ and $\Pi_{\C,h}$ interchanged for $h=1,2$, and let $\widetilde\nu$ be $\nu$ with the same interchange applied and with the counts of $+1$ and of $-1$ interchanged inside every $\nu_{L,h}$. Since $\Theta$ is invariant under both interchanges, $\widetilde\nu$ is again admissible and $\widetilde\Pi$ is compatible with it, and it leaves every $\nu_{\A,h}(0)$ unchanged. Set
\begin{equation}\label{eq:Y}
\YY(\Pi,\nu)=\XX(\widetilde\Pi,\widetilde\nu),
\end{equation}
so that a $Y\in\YY(\Pi,\nu)$ takes values in $\{0,1\}$ on $\Pi_{\A,h}$, is constant $1$ on $\Pi_{\B,h}$, and has type $\big(\nu_{\C,h}(+1),\nu_{\C,h}(0),\nu_{\C,h}(-1)\big)$ on $\Pi_{\C,h}$. Finally, let $\XX^h$ and $\YY^h$ be the restrictions to $H_h$ of the members of $\XX(\Pi,\nu)$ and of $\YY(\Pi,\nu)$; as $\widetilde\Pi$ has the same halves as $\Pi$, $\XX(\Pi,\nu)=\XX^1\times\XX^2$ and $\YY(\Pi,\nu)=\YY^1\times\YY^2$. The filtering of these lists by a random prime, in steps \ref{step:prime} to \ref{step:hash} below, is the representation technique of \cite{HJ10}, in the form used in \cite{MNPW19,BCJ11,RW26}.

\smallskip
\noindent\fbox{\begin{minipage}{0.945\textwidth}\small
\textbf{\textsc{Central}$(x,z)$.} Repeat $q(n)$ times: draw a uniform $\varepsilon\in\{\pm1\}^n$ and set $y_i:=\varepsilon_ix_i$. For each admissible layout $\nu$, repeat
\[ q(n)\Big\lceil 2^{\,n\left(\bh(p)+1-p-\alpha\log_23-2\beta\etafix\right)}\Big\rceil \ \text{times:} \]
\begin{enumerate}[label=(\roman*), ref=(\roman*)]
\item\label{step:partition} draw a uniform ordered partition $\Pi$ of $[n]$ compatible with $\nu$;
\item\label{step:prime} draw a uniform prime $P\in[K,2K]$ and a uniform residue $\rho$ modulo $P$;
\item\label{step:bucket} for $h=1,2$, bucket $\XX^h$ and $\YY^h$ by $\ip{\cdot}{y|_{H_h}}\bmod P$, storing only nonempty buckets;
\item\label{step:count} compute $N_{\XX}=\sum_r|\XX^1_r|\,|\XX^2_{(\rho-r)\bmod P}|$ and $N_{\YY}$ likewise, and \textbf{abort} the iteration if $\max(N_\XX,N_\YY)>q(n)|\XX(\Pi,\nu)|/K$;
\item\label{step:hash} materialise $\XX_\rho=\{X\in\XX(\Pi,\nu):\ip{X}{y}\equiv\rho\}$ by joining complementary nonempty buckets, and hash its members by the exact integer $\ip{X}{y}$;
\item\label{step:query} materialise $\YY_\rho$ likewise and query it against the table; on a hit $(X,Y)$ return the vector $\big(\varepsilon_i(X_i-Y_i)\big)_i$.
\end{enumerate}
\end{minipage}}

\section{Correctness and running time}\label{sec:analysis}

\subsection{One-sided error}

\begin{lemma}\label{lem:sound}
Let $\nu$ be admissible and $\Pi$ compatible with it. Then $X-Y\in\{-1,0,1\}^n$ for all $X\in\XX(\Pi,\nu)$ and $Y\in\YY(\Pi,\nu)$, and $\XX(\Pi,\nu)\cap\YY(\Pi,\nu)=\varnothing$.
\end{lemma}

\begin{proof}
On $\Pi_{\A,h}$ the third entry of the prescribed type is $0$ for both sets, so $X_i,Y_i\in\{0,1\}$. On $\Pi_{\B,h}$ we have $Y_i=1$ and $X_i\in\{0,1,2\}$, and on $\Pi_{\C,h}$ the reverse, so $X_i-Y_i\in\{-1,0,1\}$ in each case. For the second claim, every $Y$ is constant $1$ on $\Pi_{\B,1}$, whereas every $X$ has $\nu_{\B,1}(-1)+\nu_{\B,1}(+1)\ge\tfrac{4\beta n}{9}-8$ entries there that are not $1$, positive for large $n$ because $\beta\ge\tfrac3{28}$ on $[\tfrac17,\tfrac27]$.
\end{proof}

\begin{proposition}\label{prop:onesided}
\textsc{Core}$(x)$ returns either \textsf{NO} or a nonzero $D\in\{-1,0,1\}^n$ with $\ip{D}{x}=0$, and no branch can propose a vector that fails the check it performs. Consequently ESS never outputs a non-solution.
\end{proposition}

\begin{proof}
\textsc{Core} verifies $D\ne0$ and $\ip{D}{x}=0$ before returning, and ESS verifies the same two conditions against the original instance before any output, so its output is a solution by \hyperref[lem:enc-a]{Lemma~\ref*{lem:enc-a}}, whatever the reduction and the branches do. For the second claim: a hit in \textsc{Central} has $X\ne Y$ and $X-Y\in\{-1,0,1\}^n$ by \cref{lem:sound}, and $\ip{X-Y}{y}=0$ by construction, so $(\varepsilon_i(X_i-Y_i))_i$ is nonzero and annihilates $x$ by \cref{lem:sign}; a hit in \textsc{Mitm} is by construction a nonzero $D\in\{-1,0,1\}^n$ with $\ip{D}{x}=0$; and a hit in \textsc{Collide} returns the vector attached by \hyperref[lem:enc-a]{Lemma~\ref*{lem:enc-a}} to two distinct subsets of equal sum.
\end{proof}

\subsection{A balanced solution of minimum support, and a layout matching it}\label{ssec:balanced}

Assume $x$ has a solution and fix a nonzero $D^\star\in\{-1,0,1\}^n$ of minimum support with $\ip{D^\star}{x}=0$. Let $z$ be its number of zero coordinates and $p=z/n$. Until \cref{ssec:other} we assume $p\in(\tfrac17,\tfrac27)$, the range in which \textsc{Central} is invoked, so that \eqref{eq:par} and \eqref{eq:theta} are defined. For a uniform $\varepsilon\in\{\pm1\}^n$ the signs of $\varepsilon_iD^\star_i$ on $\supp(D^\star)$ are independent and uniform, so with probability at least $c/\sqrt n$ for an absolute constant $c>0$ the vector $D:=(\varepsilon_iD^\star_i)_i$ satisfies
\begin{equation}\label{eq:bal}
\Big|\,\big|\{i:D_i=1\}\big|-\big|\{i:D_i=-1\}\big|\,\Big|\le1 .
\end{equation}
By \cref{lem:sign}, $D$ then annihilates $y=(\varepsilon_ix_i)_i$ and its support is again of minimum size among the nonzero vectors of $\{-1,0,1\}^n$ annihilating $y$. The $q(n)$ outer repetitions of \textsc{Central} realise \eqref{eq:bal} except with probability $2^{-\Omega(n)}$. In the rest of \cref{sec:analysis}, $D$ and $y$ denote these, and $n_d$ denotes $|\{i:D_i=d\}|$, so that $n_0=z$ and $|n_d-n\Theta_{d\bullet}|\le1$ for every $d$, where $\Theta_{d\bullet}$ is the $d$-th row sum of $\Theta$.

\begin{lemma}\label{lem:layout}
Some admissible layout $\nu$ satisfies $\sum_{L,h}\nu_{L,h}(d)=n_d$ for every $d\in\{-1,0,1\}$.
\end{lemma}

\begin{proof}
Choose integers $c_L$ summing to $n$ with $|c_L-n\Theta_{\bullet L}|\le2$, where $\Theta_{\bullet L}$ is the $L$-th column sum. Put $r_d=n_d-n\Theta_{d\bullet}$ and $s_L=c_L-n\Theta_{\bullet L}$, so that $|r_d|\le1$, $|s_L|\le2$ and $\sum_dr_d=\sum_Ls_L=0$. Then $M^{*}_{dL}=n\Theta(d,L)+\tfrac13r_d+\tfrac13s_L$ has row sums $(n_d)$ and column sums $(c_L)$, satisfies $\|M^{*}-n\Theta\|_\infty\le1$, and is nonnegative because every entry of $n\Theta$ is at least $\beta n/9$. The polytope of real arrays with those row and column sums and with $\lfloor M^{*}\rfloor\le M\le\lceil M^{*}\rceil$ is nonempty, its constraint matrix is totally unimodular and its data are integral, so it contains an integral point $M$ \cite[Chapter 19]{Schrijver86}, which satisfies $\|M-n\Theta\|_\infty\le2$. Split each $M_{dL}$ into two halves differing by at most $1$ to obtain $\nu_{L,h}$; these are within $\tfrac32$ of $\tfrac n2\Theta(d,L)$, so \eqref{eq:adm} holds and $\sum_{L,h}\nu_{L,h}(0)=n_0=z$.
\end{proof}

\subsection{The probability that a partition is good}

For a probability vector $\pi$ put $\Ent(\pi)=-\sum_\sigma\pi(\sigma)\log_2\pi(\sigma)$, in bits, and for a type $\nu$ of length $N\ge1$ write $\Ent(\nu)$ for $\Ent(\nu/N)$; thus $\bh$ is the case of two symbols. We use throughout the estimates
\begin{equation}\label{eq:tc}
\frac{2^{N\Ent(\nu)}}{(N+1)^{|\Sigma|}}\ \le\ \binom{N}{\nu}\ \le\ 2^{N\Ent(\nu)},\qquad\quad \binom{N}{\nu}\Big/\binom{N'}{\nu'}\in\Big[\frac{1}{\poly},\ \poly\Big],
\end{equation}
the first being the method of types \cite[Chapter 11]{CoverThomas06} and the second holding whenever $|N-N'|$ and $\|\nu-\nu'\|_\infty$ are $O(1)$, since altering one entry by $1$ changes a multinomial coefficient by a factor of at most $N+1$. Splitting each symbol of a type into two nearly equal halves adds one bit to $\Ent$ up to $O(\log n/n)$, hence a factor $\poly$ in \eqref{eq:tc}; we use this for the six blocks against the three labels.

Decomposing the entropy of the array $\Theta$ first along its columns and then along its rows,
\begin{equation}\label{eq:chain}
\begin{aligned}
-\sum_{d,L}\Theta(d,L)\log_2\Theta(d,L) &=\Ent(\alpha,\beta,\beta)+\sum_{L}\Theta_{\bullet L}\,\Ent\biggl(\frac{\Theta(\cdot,L)}{\Theta_{\bullet L}}\biggr)\\ &=\Ent\Big(\tfrac{1-p}2,\,p,\,\tfrac{1-p}2\Big)+\sum_{d}\Theta_{d\bullet}\,\Ent\biggl(\frac{\Theta(d,\cdot)}{\Theta_{d\bullet}}\biggr).
\end{aligned}
\end{equation}
The column $\A$ of $\Theta$ is uniform on the three rows and the columns $\B$ and $\C$ are proportional to $(\tfrac49,\tfrac19,\tfrac49)$, whose entropy is $\log_29-\tfrac{16}9=\etafix$; also $\Ent\Big(\tfrac{1-p}2,p,\tfrac{1-p}2\Big)=\bh(p)+1-p$. So the column-conditional term $\sum_{L}\Theta_{\bullet L}\Ent(\Theta(\cdot,L)/\Theta_{\bullet L})$ in \eqref{eq:chain} equals $\alpha\log_23+2\beta\etafix$ and \eqref{eq:chain} rearranges to
\begin{equation}\label{eq:filterexp}
\Ent(\alpha,\beta,\beta)-\sum_{d}\Theta_{d\bullet}\,\Ent\biggl(\frac{\Theta(d,\cdot)}{\Theta_{d\bullet}}\biggr)\;=\;\bh(p)+1-p-\alpha\log_23-2\beta\etafix,
\end{equation}
which is the exponent appearing in the repetition count of \textsc{Central}.

Call a partition $\Pi$ compatible with $\nu$ \emph{good} if for all $L$ and $h$ the restriction of $D$ to $\Pi_{L,h}$ has type $\nu_{L,h}$.

\begin{proposition}\label{prop:good}
Let $\nu$ be as in \cref{lem:layout} and let $\Pi$ be uniform among the partitions compatible with $\nu$. Then $\Pr[\Pi\ \text{is good}]\ \ge\ 2^{-n\left(\bh(p)+1-p-\alpha\log_23-2\beta\etafix\right)}\big/\poly$.
\end{proposition}

\begin{proof}
A good $\Pi$ is obtained by distributing, independently for each $d$, the $n_d$ coordinates with $D_i=d$ among the six blocks with the prescribed counts $\nu_{L,h}(d)$, so the probability in question equals $\prod_d\binom{n_d}{(\nu_{L,h}(d))_{L,h}}$ divided by $\binom{n}{(|\Pi_{L,h}|)_{L,h}}$. By \eqref{eq:adm} and \eqref{eq:tc} the numerator is
\[ 2^{\,n\left(\sum_d\Theta_{d\bullet}\Ent\left(\Theta(d,\cdot)/\Theta_{d\bullet}\right)+1\right)}\big/\poly \]
and the denominator is at most $2^{\,n(\Ent(\alpha,\beta,\beta)+1)}\poly$, the extra bit in each being the split of each label into two nearly equal halves. The claim now follows from \eqref{eq:filterexp}.
\end{proof}

\subsection{Representations of $D$, and the distinctness of their sums}

\begin{lemma}\label{lem:reps}
Let $\Pi$ be good and put $\PP=\{(X,Y)\in\XX(\Pi,\nu)\times\YY(\Pi,\nu):X-Y=D\}$ and $F=\{i\in\Pi_{\A,1}\cup\Pi_{\A,2}:D_i=0\}$. Then $(X,Y)\mapsto\{i\in F:X_i=0\}$ is a bijection from $\PP$ onto the family of subsets $G\subseteq F$ with $|G\cap\Pi_{\A,h}|=\lfloor\nu_{\A,h}(0)/2\rfloor$ for $h=1,2$. Consequently $|\PP|=K=2^{\kappa n}/\poly$.
\end{lemma}

\begin{proof}
On $\Pi_{\B,h}$ we have $Y\equiv1$, so $X=D+1$ there, and its type is $\big(\nu_{\B,h}(-1),\nu_{\B,h}(0),\nu_{\B,h}(+1)\big)$ as required by \eqref{eq:lists}; the choice is forced. On $\Pi_{\C,h}$ we have $X\equiv1$, so $Y=1-D$ there, of type $\big(\nu_{\C,h}(+1),\nu_{\C,h}(0),\nu_{\C,h}(-1)\big)$ as required by \eqref{eq:Y}; again forced. On $\Pi_{\A,h}$ both $X_i$ and $Y_i$ lie in $\{0,1\}$ with $X_i-Y_i=D_i$, so $D_i=-1$ forces $(X_i,Y_i)=(0,1)$, $D_i=+1$ forces $(1,0)$, and $D_i=0$ leaves the free choice between $(0,0)$ and $(1,1)$. Writing $G_h$ for the free coordinates of $\Pi_{\A,h}$ assigned $(0,0)$, the number of those with $X_i=0$ is $\nu_{\A,h}(-1)+|G_h|$, which agrees with \eqref{eq:lists} exactly when $|G_h|=\lfloor\nu_{\A,h}(0)/2\rfloor$; and in that case $Y$ has the type required by \eqref{eq:Y} on $\Pi_{\A,h}$, by the same count with $+1$ and $-1$ interchanged. Finally $\binom{\nu_{\A,h}(0)}{\lfloor\nu_{\A,h}(0)/2\rfloor}=2^{\nu_{\A,h}(0)}/\poly$ by \eqref{eq:tc}, and $\nu_{\A,1}(0)+\nu_{\A,2}(0)=\kappa n\pm8$ by \eqref{eq:adm}.
\end{proof}

\begin{lemma}\label{lem:mix}
Let $\Pi$ be good. Every $(X,Y)\in\PP$ satisfies $\ip{X}{y}=\ip{Y}{y}$, and the $K$ integers $\ip{X}{y}$ for $(X,Y)\in\PP$ are pairwise distinct.
\end{lemma}

\begin{proof}
$\ip{X}{y}-\ip{Y}{y}=\ip{D}{y}=0$. Suppose $(X,Y)$ and $(X',Y')$ in $\PP$ have $\ip{X}{y}=\ip{X'}{y}$. By \cref{lem:reps} they differ on a nonempty subset of $F$, so $X-X'$ is a nonzero vector of $\{-1,0,1\}^n$ with $\ip{X-X'}{y}=0$ and support contained in $F$. But $|F|=\nu_{\A,1}(0)+\nu_{\A,2}(0)\le\kappa n+8$, while $1-p-\kappa=\tfrac{6-6p-9p+1}{6}=\tfrac{7-15p}{6}\ge\tfrac16$ on $\big[\tfrac17,\tfrac27\big]$, so for large $n$ the support of $X-X'$ has fewer than $(1-p)n=|\supp(D)|$ elements, contradicting minimality.
\end{proof}

Writing a solution as a difference of two vectors and drawing its representations from its zero coordinates are due to \cite{MNPW19}, letting both vectors take values in $\{0,1,2\}$ to \cite{RW26}, and the argument of \cref{lem:mix} is that of \cite[Lemma 7.2]{RW26}; what \cref{lem:reps,lem:mix} add is that the free coordinates are confined to the label $\A$, so their density $\kappa$ is smaller than $1-p$.

\subsection{The cost of one iteration}

\begin{proposition}\label{prop:cost}
Let $\nu$ be admissible and $\Pi$ compatible with it. Then $|\XX(\Pi,\nu)|$ and $|\YY(\Pi,\nu)|$ are both $2^{\,n(\alpha+\beta\etafix)}/\poly$, and an iteration of \textsc{Central} that does not abort runs in time and space $O\big(\big(2^{\,n(\alpha+\beta\etafix)/2}+2^{\,n(\alpha+\beta\etafix-\kappa)}\big)\poly\big)=O\big(2^{\,n(\alpha+\beta\etafix-\kappa)}\poly\big)$.
\end{proposition}

\begin{proof}
By \eqref{eq:lists}, $|\XX(\Pi,\nu)|$ is the product over $h$ of the two multinomial coefficients belonging to the blocks $\Pi_{\A,h}$ and $\Pi_{\B,h}$. By \eqref{eq:adm} the first has length $\tfrac{\alpha n}2\pm O(1)$ and type within $O(1)$ of $\tfrac{\alpha n}{4}(1,1,0)$, hence equals $2^{\alpha n/2}/\poly$ by \eqref{eq:tc}; the second has length $\tfrac{\beta n}{2}\pm O(1)$ and type within $O(1)$ of $\tfrac{\beta n}{2}\big(\tfrac49,\tfrac19,\tfrac49\big)$, hence equals $2^{\beta\etafix n/2}/\poly$. Multiplying the four factors gives $2^{n(\alpha+\beta\etafix)}/\poly$. The same computation applies to $\widetilde\nu$ and $\widetilde\Pi$, which are admissible and compatible, so $|\YY(\Pi,\nu)|$ obeys the same estimate.

Steps \ref{step:partition} to \ref{step:count} enumerate the four sets $\XX^h,\YY^h$, each of size $2^{n(\alpha+\beta\etafix)/2}/\poly$ by the same computation applied to one half, and sort them into sparse buckets; step \ref{step:count} needs only bucket sizes, so it precedes any materialisation. Steps \ref{step:hash} and \ref{step:query} emit each surviving vector once, so together they cost $O\big(\big(2^{n(\alpha+\beta\etafix)/2}+N_\XX+N_\YY\big)\poly\big)$, and on an iteration that does not abort $N_\XX$ and $N_\YY$ are at most $q(n)|\XX(\Pi,\nu)|/K=\poly\cdot2^{n(\alpha+\beta\etafix-\kappa)}$ by \cref{lem:reps}. The first term is dominated by the second because $\alpha+\beta\etafix\ge\alpha=3\kappa$, whence $\kappa\le\tfrac13(\alpha+\beta\etafix)$ and therefore $\tfrac12(\alpha+\beta\etafix)\le\alpha+\beta\etafix-\kappa$.
\end{proof}

\subsection{The modular filter}

\begin{lemma}\label{lem:prime}
Let $\sigma_1,\dots,\sigma_K$ be distinct integers of absolute value below $2^{\ell}$ for some $\ell\le\poly$, and assume $K=2^{\Theta(n)}$. For a uniform prime $P\in[K,2K]$, with probability at least $\tfrac12$ the residues $\sigma_j\bmod P$ take at least $K/\poly$ distinct values.
\end{lemma}

\begin{proof}
Each nonzero difference $\sigma_j-\sigma_{j'}$ has at most $(\ell+1)/\log_2K=\poly$ prime divisors that are at least $K$, while $[K,2K]$ contains at least $K/(2\ln K)$ primes by the prime number theorem, so a fixed pair collides with probability at most $\poly/K$. Hence the number $Z$ of colliding pairs has expectation at most $K\poly$, and $Z\le2K\poly$ with probability at least $\tfrac12$. If the residues fall into $M$ classes then $Z\ge K^2/(2M)-K/2$ by convexity, so $M\ge K/\poly$.
\end{proof}

\begin{proposition}\label{prop:trial}
Condition on $\Pi$ being good. With probability at least $1/\poly$ over the choice of $P$ and $\rho$, the iteration does not abort and returns a solution.
\end{proposition}

\begin{proof}
By \cref{lem:mix} the $K$ integers $\ip{X}{y}$, $(X,Y)\in\PP$, are distinct, and they are bounded in absolute value by $2^{O(n)}$; moreover $K=2^{\kappa n}/\poly$ with $\kappa\ge\tfrac1{21}$ on $[\tfrac17,\tfrac27]$. So \cref{lem:prime} applies, and with probability at least $\tfrac12$ the set $\mathcal R$ of residues modulo $P$ of these $K$ integers has $|\mathcal R|\ge K/\poly$. Assume this from now on.

For a residue $r$ let $N_\XX(r)$ be the number of $X\in\XX(\Pi,\nu)$ with $\ip{X}{y}\equiv r$, and define $N_\YY(r)$ likewise. Summing over all $r$ gives $|\XX(\Pi,\nu)|+|\YY(\Pi,\nu)|$, which is at most $\poly\cdot|\XX(\Pi,\nu)|$ by \cref{prop:cost}. Averaging over $\mathcal R$ and applying Markov's inequality, at least half of the residues in $\mathcal R$ satisfy $\max\big(N_\XX(r),N_\YY(r)\big)\le\poly|\XX(\Pi,\nu)|/K$, and for $c_0$ large enough such a residue passes the test in step \ref{step:count}, which computes $N_\XX(\rho)$ and $N_\YY(\rho)$ exactly. If moreover $\rho\in\mathcal R$, choose $(X,Y)\in\PP$ with $\ip{X}{y}\equiv\rho$; since $\ip{Y}{y}=\ip{X}{y}$ we get $X\in\XX_\rho$ and $Y\in\YY_\rho$ with the same exact sum, so step \ref{step:query} reports a hit, which is a solution by \cref{prop:onesided} whether or not it is this particular pair. A uniform $\rho$ satisfies both requirements with probability at least $\tfrac12|\mathcal R|/P\ge1/\poly$; the sampler of step \ref{step:prime} gives up with probability $2^{-\Omega(n)}$, absorbed into that bound.
\end{proof}

\subsection{The exponent of the central branch}

\begin{theorem}\label{thm:E}
For $p\in(\tfrac17,\tfrac27)$, with $\alpha,\beta,\kappa$ as in \eqref{eq:par}, $\big(\bh(p)+1-p-\alpha\log_23-2\beta\etafix\big)+\big(\alpha+\beta\etafix-\kappa\big)=\bh(p)+2(1-p)-\log_23$. As a function of $p$ on $[0,1]$, the right-hand side is strictly concave, is maximised only at $p=\tfrac15$, and its value there is $\log_2\tfrac53$.
\end{theorem}

\begin{proof}
Since $\kappa=\alpha/3$, the left-hand side is $\bh(p)+1-p+\tfrac23\alpha-\alpha\log_23-\beta\etafix$. Substituting $\alpha=1-2\beta$ turns $-\alpha\log_23$ into $-\log_23+2\beta\log_23$, and
\[ 2\beta\log_23-\beta\etafix=\beta\Big(2\log_23-2\log_23+\tfrac{16}9\Big)=\frac{16\beta}{9},\qquad \frac{2\alpha}{3}+\frac{16\beta}{9}=\frac{9p-1}{3}+\frac{4(1-3p)}{3}=1-p , \]
so the left-hand side equals $\bh(p)+(1-p)-\log_23+(1-p)$. Its derivative in $p$ is $\log_2\frac{1-p}{p}-2$, which vanishes only when $1-p=4p$, and its second derivative is $-1/\big(p(1-p)\ln2\big)<0$. Finally $\bh\big(\tfrac15\big)=\tfrac15\log_25+\tfrac45\log_2\tfrac54=\log_25-\tfrac85$, so the value at $p=\tfrac15$ is $\log_25-\tfrac85+\tfrac85-\log_23=\log_2\tfrac53$.
\end{proof}

\begin{proposition}\label{prop:central}
If $x$ has a solution whose support is of minimum size and whose density of zero coordinates is $p=z/n\in(\tfrac17,\tfrac27)$, then \textsc{Central}$(x,z)$ returns a solution with probability $1-2^{-\Omega(n)}$ and runs in time and space
\[ O\big(2^{\,n(\bh(p)+2(1-p)-\log_23)}\poly\big)=O\big((5/3)^n\poly\big). \]
\end{proposition}

\begin{proof}
The three nested loops perform
\[ q(n)\cdot O(1)\cdot q(n)\Big\lceil2^{\,n\left(\bh(p)+1-p-\alpha\log_23-2\beta\etafix\right)}\Big\rceil \]
iterations, each costing $O\big(2^{n(\alpha+\beta\etafix-\kappa)}\poly\big)$ by \cref{prop:cost}, an iteration that aborts costing strictly less; the running time claim is then Theorem \ref{thm:E}. For correctness, except with probability $2^{-\Omega(n)}$ some outer repetition produces $\varepsilon$ satisfying \eqref{eq:bal}, and \cref{lem:layout} supplies a matching admissible layout among those enumerated. Each subsequent iteration independently draws a good $\Pi$ with probability at least $2^{-n(\bh(p)+1-p-\alpha\log_23-2\beta\etafix)}/\poly$ by \cref{prop:good} and then succeeds with probability at least $1/\poly$ by \cref{prop:trial}, so for $c_0$ large enough the failure probability over all the iterations is $2^{-\Omega(n)}$.
\end{proof}

\subsection{The other two branches, and the main theorem}\label{ssec:other}

\begin{proposition}\label{prop:outer}
Let $x$ have a solution whose support is of minimum size, with $z$ zero coordinates and $p=z/n$. Then \textsc{Mitm}$(x,z)$ runs in time and space $O\big(2^{\,n(\bh(p)+1-p)/2}\poly\big)$, and, provided $w=\min\{z,n-z-1\}=\Omega(n)$, \textsc{Collide}$(x,z)$ runs in time and space $O\big(2^{\,n\max\{1/2,\,1-p,\,p\}}\poly\big)$. Each returns a solution with probability $1-2^{-\Omega(n)}$.
\end{proposition}

\begin{proof}
For \textsc{Mitm}, a uniform halving splits each of the three classes of $D^\star$ evenly up to $1$ with probability at least $1/\poly$, and then the types of $D^\star|_{H_1}$ and $D^\star|_{H_2}$ agree entrywise up to $1$ and their numbers of zeros sum to $z$, so this pair is enumerated and the two halves meet in the hash table; the $q(n)$ repetitions amplify this to $1-2^{-\Omega(n)}$. For the running time, write $N=|H_h|$ and split the multinomial coefficient counting the set enumerated for $\nu^h$ as $\binom{N}{\nu^h(0)}\binom{N-\nu^h(0)}{\nu^h(+1)}\le\binom{N}{\nu^h(0)}2^{\,N-\nu^h(0)}$. Here $N=\tfrac n2\pm1$ and $\nu^h(0)=\tfrac z2\pm1$, so the first factor is at most $2^{\,n\bh(p)/2}\poly$ by \eqref{eq:tc} and the second is $2^{\,(1-p)n/2+O(1)}$, whence the product is $2^{\,n(\bh(p)+1-p)/2}\poly$.

For \textsc{Collide}, let $(A,B)$ be the pair of sets corresponding to $D^\star$ and let $W$ consist of $w$ of its zero coordinates. For every $T\subseteq W$ the sets $A\cup T$ and $B\cup T$ are distinct and have equal sums, and the $2^{w}$ integers $\sum_{i\in A}x_i+\sum_{i\in T}x_i$ are pairwise distinct: two coinciding ones would give, by \hyperref[lem:enc-a]{Lemma~\ref*{lem:enc-a}}, a solution of support at most $w\le n-z-1<|\supp(D^\star)|$, contradicting minimality. These $2^{w}$ integers play the role of the hidden sums in \cref{prop:trial}, and $w=\Omega(n)$ by hypothesis, so \cref{lem:prime} with $K=2^{w}$ applies: with probability at least $\tfrac12$ they occupy at least $2^{w}/\poly$ residues modulo $P$. Averaging the $2^{n}$ subsets of $[n]$ over those residues and applying Markov's inequality leaves at least half of them carrying at most $2^{\,n-w}\poly$ subsets, which for $c_0$ large enough passes the abort test; a uniform $\rho$ is such a residue with probability at least $1/\poly$ (its sampler gives up with probability $2^{-\Omega(n)}$, absorbed into this bound), and then some $T\subseteq W$ satisfies $\sum_{i\in A}x_i+\sum_{i\in T}x_i\equiv\rho$, so that $A\cup T$ and $B\cup T$ both lie in the materialised set and share their exact sum. The cost is $O\big((2^{n/2}+2^{\,n-w})\poly\big)$ and $n-w=\max\{n-z,\,z+1\}$.
\end{proof}

\begin{theorem}\label{thm:main}
Let $x\in\mathbb Z^n$ have entries of at most $m\le2^n$ bits. Then ESS$(x)$ never outputs a non-solution, outputs a solution with probability $1-2^{-\Omega(n)}$ whenever one exists, and runs in time $O\big((5/3)^n\poly+n^2m\big)$ and space $O\big((5/3)^n\poly+nm\big)$.
\end{theorem}

\begin{proof}
The first claim is \cref{prop:onesided}. Suppose $x$ has a solution. Then $\widetilde x$ has one by \hyperref[lem:reduce-a]{Lemma~\ref*{lem:reduce-a}}; fix one of minimum support, so that \textsc{Core} reaches its number $z$ of zero coordinates and invokes the branch selected by $p=z/n$, which succeeds with probability $1-2^{-\Omega(n)}$ by \cref{prop:central,prop:outer}. By \hyperref[lem:reduce-b]{Lemma~\ref*{lem:reduce-b}} the restriction to $[n]$ of the vector returned then passes the check of ESS, except on the events, of total probability $2^{-\Omega(n)}$, that the sampler of \cref{lem:reduce} gives up, whereupon \textsc{ESS} outputs \textsf{NO}, or that its prime is bad. For the running time, the derivative of $\tfrac12(\bh(p)+1-p)$ is $\tfrac12\big(\log_2\tfrac{1-p}p-1\big)$, so it increases on $[0,\tfrac13]$ and decreases on $[\tfrac13,1]$, and we obtain
\[ \begin{array}{lll} p\le\tfrac17: & \textsc{Mitm}: & \tfrac12(\bh(p)+1-p)\le\tfrac12\big(\bh(\tfrac17)+\tfrac67\big)=0.7244078\ldots\\[2pt] \tfrac17<p<\tfrac27: & \textsc{Central}: & \bh(p)+2(1-p)-\log_23\le\log_2\tfrac53=0.7369655\ldots\\[2pt] \tfrac27\le p<\tfrac23: & \textsc{Collide}: & \max\{\tfrac12,1-p,p\}\le\max\{\tfrac12,\tfrac57,\tfrac23\}=\tfrac57=0.7142857\ldots\\[2pt] p\ge\tfrac23: & \textsc{Mitm}: & \tfrac12(\bh(p)+1-p)\le\tfrac12\big(\bh(\tfrac23)+\tfrac13\big)=0.6258145\ldots \end{array} \]
using Theorem \ref{thm:E} on the second line. On the third line the hypothesis of \cref{prop:outer} holds, since there $w\ge\min\{\tfrac27n,\tfrac13n\}-1=\Omega(n)$. Every branch of \textsc{Core} is therefore $O\big((5/3)^n\poly\big)$; the $n$ values of $z$ cost a further factor of $n$, the $O(\log n)$ entries appended by \cref{lem:reduce} a factor $(5/3)^{O(\log n)}=\poly$, and that lemma an additive $O(n^2m+\poly)$, whose $\poly$ the exponential term absorbs. Finally \textsc{Core} uses $O\big((5/3)^n\poly\big)$ space, which absorbs the $O(n^2)$ bits occupied by $\widetilde x$, while the verification against $x$ needs $O(m+\log n)$ bits, within the $O(nm)$ of the instance itself.
\end{proof}

\end{document}